\documentclass[journal]{IEEEtran}
\usepackage{amsmath,amsfonts}
\usepackage{algorithmic}
\usepackage{algorithm}
\usepackage{array}
\usepackage{textcomp}
\usepackage{stfloats}
\usepackage{url}
\usepackage{verbatim}
\usepackage{graphicx}
\usepackage{cite}
\usepackage{amssymb}

\newtheorem{proposition}{Proposition}
\newtheorem{theorem}{Theorem}

\newcounter{prob}
\renewcommand{\theprob}{P\arabic{prob}}

\usepackage[
    caption=false,
    font=footnotesize,
    labelfont=rm,
    textfont=rm
]{subfig}
\begin{document}

\title{UAV Rotatable Antenna under Rotation Delay}

\author{Yinmei~Xie,~\IEEEmembership{Graduate Student Member,~IEEE}, Chiya~Zhang,~\IEEEmembership{Member,~IEEE},
Chunlong~He,~\IEEEmembership{Member,~IEEE},
Weizheng~Zhang,~\IEEEmembership{Member,~IEEE},
Wei~Wang,~\IEEEmembership{Senior Member,~IEEE}
}



\maketitle

\begin{abstract}
Rotatable antennas (RAs) can improve wireless link performance by adjusting their boresight directions. 
However, existing RA studies generally assume instantaneous antenna reconfiguration, whereas mechanical rotation incurs non-negligible delay in practical deployment.
To address this issue, we develop a rotation-delay-aware communication model for finite communication slots, where mechanical rotation directly reduces the time available for data transmission. This model reveals a fundamental tradeoff between the communication gain achieved through RA reorientation and the associated transmission-time loss. Accordingly, we formulate a rotation-delay-aware joint design of RA orientations and transmit beamforming to maximize the minimum effective transmission rate.
For the single-user system, under a sufficient condition for strict concavity, the global optimum falls into one of three regimes: no rotation, partial rotation, and full rotation.
A closed-form approximation is derived using a second-order expansion. 
For the multi-user system, we separate rotation-budget selection from fixed-budget max-min signal-to-interference-plus-noise ratio (SINR) optimization.
An analytic-gradient zero-forcing (AG-ZF) algorithm is then developed for low-complexity joint RA orientation and rotation-budget optimization.
The results show that explicitly accounting for mechanical rotation delay is essential for finite-slot RA-enabled UAV communications, and that the proposed designs effectively balance reconfiguration gain and transmission-time loss while achieving favorable effective-rate performance with low computational complexity.

\end{abstract}



\section{Introduction}
\IEEEPARstart{U}{nmanned} aerial vehicles (UAVs) have been widely used in inspection, logistics, mapping, emergency communications, and cooperative sensing due to their high mobility, flexible deployment, and three-dimensional coverage capabilities~\cite{ZengYongUAV,HayatUAV,MozaffariUAV,ZengUAVTutorial}. 
In UAV downlink communications, a mobile UAV can serve multiple spatially distributed users~\cite{XuUAV}. As the UAV moves, the relative geometry and propagation directions of the links between the UAV and the users can change over time. 
In mission-driven applications, UAV trajectories may be predetermined~\cite{PredeterminedTrajectory}, motivating antenna-side adaptation without altering the flight path~\cite{GuoBeamforming}.
Conventional antenna arrays can adapt their beamforming weights to the instantaneous channel~\cite{Khallilibeamforming,Lyubeamforming}, but their element positions and orientations remain fixed, limiting their spatial reconfiguration capability.
When users are located in low-gain antenna directions or their channels are highly correlated, beamforming alone may provide limited performance improvement. 
To improve channel adaptability, several reconfigurable antenna architectures have recently been investigated, including movable antennas (MAs), six-dimensional movable antennas (6DMAs), and rotatable antennas (RAs)~\cite{MA3D,HuaHaocheng6DMA,ZhengBeixiongRA,ZhengBeixiongRASurvey,PeiFengchengOptRA,RAModel,XiangLinEnergy}. 
Such architectures can reconfigure the wireless channel by adjusting antenna positions or orientations. 
They can improve link gains and reduce inter-user channel correlation.
Among these architectures, RAs adjust the boresight directions of antenna elements while keeping their physical positions fixed. 
Hence, RAs can change the effective antenna gains toward different propagation directions without requiring additional space for antenna translation. 
This feature makes RAs attractive for size, weight, and power (SWaP)-constrained UAV platforms~\cite{ZhangXuzhongRAUAV,RAUAV,RAUAV3,RAUAV2}.
RA reconfiguration can be implemented using mechanical or electronic mechanisms. Electronic RAs offer fast reconfiguration, but their achievable steering states are generally constrained by the underlying tunable hardware~\cite{Chiao1999MEMS,ZhengBeixiongTutorial}. In contrast, mechanical RAs directly rotate the physical antenna boresight and provide continuous wide-angle orientation control~\cite{ZhengBeixiongTutorial}. Such flexibility is well suited to UAV downlink communications with spatially distributed users and time-varying link directions, enabling directional antenna gains to be more effectively exploited through boresight adjustment. Therefore, this work focuses on mechanically driven RAs.

Most existing RA studies optimize communication performance after antenna reconfiguration without explicitly accounting for the execution time of rotation~\cite{ZhengBeixiongRA,PeiFengchengOptRA,ZhangXuzhongRAUAV,RAModel,PeiFengchengRA}.
Recent work has further considered finite rotation-speed constraints and the energy consumption of dynamic mechanical steering~\cite{XiangLinEnergy}. However, these studies do not explicitly model the reduction in available data-transmission time caused by mechanical rotation within a finite communication slot.
Related work on movable antennas (MAs) has incorporated mechanical movement delay into finite transmission blocks~\cite{WangHonghaoMA,LiQingliangMA,ZhengXiMA}, but the resulting delay-aware designs are based on antenna repositioning and are not directly applicable to RA systems.
Mechanical RA reorientation incurs non-negligible delay due to finite rotation speed~\cite{ZhengBeixiongTutorial,ViewproTraceG1}, reducing the time available for data transmission within a finite communication slot.
Hence, RA reconfiguration involves a tradeoff between communication improvement and transmission-time loss.
Optimizing only the communication performance after reconfiguration can be suboptimal when rotation time is included in the slot.

This motivates explicit consideration of mechanical rotation delay in finite-slot RA communications. Incorporating rotation delay, however, makes the RA design more involved. Since all RAs rotate in parallel, the reconfiguration delay is determined by the maximum rotation angle, which couples the orientation decisions of different antennas through a common loss of transmission time. Meanwhile, the RA orientations determine the directional channel gains and hence the multi-user channel matrix, making them strongly coupled with transmit beamforming and inter-user interference. As a result, the rotation budget, RA orientations, and transmit beamforming are tightly coupled in the original formulation, leading to a challenging nonconvex optimization problem.

Building on the above discussion, we study an RA-enabled UAV downlink communication system with mechanical rotation delay under finite communication slots, as illustrated in Fig.~\ref{fig场景图}. 
Unlike existing RA models that neglect reconfiguration delay, we explicitly account for the time consumed by antenna rotation. 
The main contributions are summarized as follows.

\begin{figure*}[!t]
	\centering
	\includegraphics[
		width=0.9\textwidth,
		clip
	]{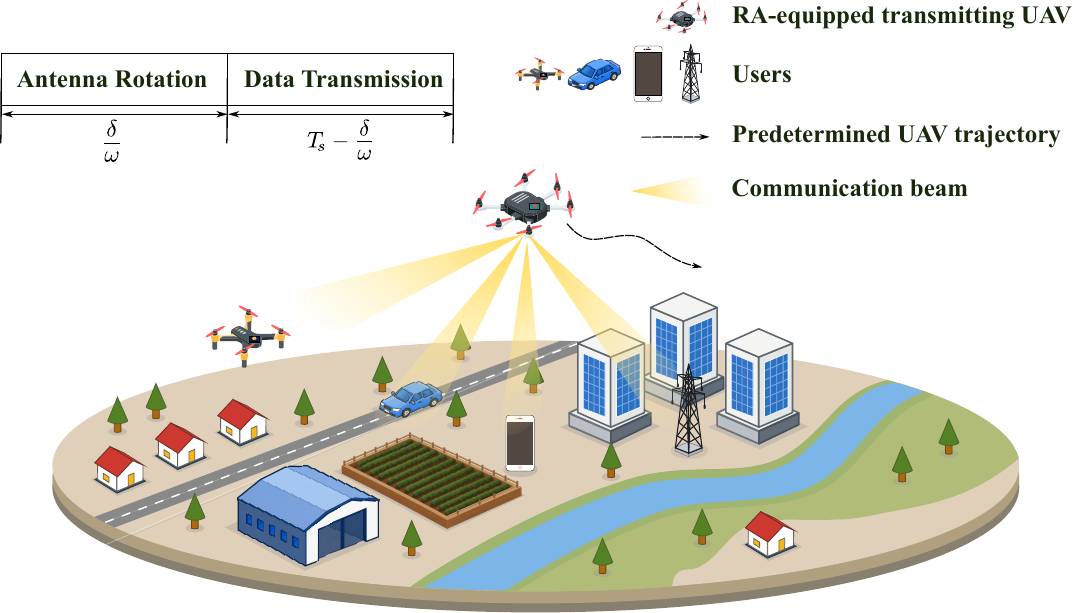}
	\caption{Illustration of the considered RA-enabled UAV downlink communication scenario.}
	\label{fig场景图}
\end{figure*}

\begin{itemize}

    \item We first develop a rotation-delay-aware communication model for mechanically driven RAs under finite communication slots.
    With parallel antenna rotation, the reconfiguration delay is determined by the maximum rotation angle among all RAs.
    Based on this model, we define the effective transmission rate and formulate a joint optimization problem for RA orientations, transmit beamforming, and the rotation budget.
    
    \item For the single-user system, we reduce the rotation-delay-aware joint design to a one-dimensional rotation problem by exploiting the optimal beamforming and fixed-budget RA pointing structures.
    Under a sufficient condition for strict concavity, the global optimum is unique and falls into three regimes: no rotation, partial rotation, and full rotation.
    We further derive a closed-form approximation through a second-order expansion, whose effective-rate performance is evaluated numerically.
    
    \item For the multi-user system, we reveal a separable structure in the rotation-delay-aware joint problem.
    The problem can be equivalently decomposed into an outer optimization over the common rotation budget and a fixed-budget max-min signal-to-interference-plus-noise ratio (SINR) optimization.
    For a given rotation budget, we impose a zero-forcing (ZF) structure so that the
    beamforming directions are determined by the RA orientations, while the
    remaining max-min fair power allocation is obtained in closed form,
    reducing the fixed-budget design to a Gram-matrix-based RA pointing problem.
    Building on this structure, we derive analytical gradients and spherical-cap feasible updates, leading to the proposed low-complexity analytic-gradient zero-forcing (AG-ZF) algorithm.
    
    \item Numerical results verify the importance of accounting for mechanical rotation delay in finite-slot UAV communications.
    The closed-form design achieves near-optimal effective-rate performance, while AG-ZF achieves favorable effective-rate performance across different multi-user settings.
    Theoretical complexity analysis further demonstrates the computational advantage of AG-ZF.
\end{itemize}

The remainder of this paper is organized as follows.
Section II presents the system model and defines the effective transmission rate with mechanical rotation delay.
Section III studies the single-user system, characterizes the optimal rotation structure, and derives an approximate closed-form solution.
Section IV considers the multi-user system, develops the AG-ZF algorithm, and analyzes its computational complexity.
Section V presents numerical results to evaluate the proposed rotation-delay-aware design and optimization algorithms.
Finally, Section VI concludes the paper.

\textit{Notation:}
Boldface uppercase and lowercase letters denote matrices and column vectors, respectively.
The superscripts $(\cdot)^T$, $(\cdot)^*$, $(\cdot)^H$, and $(\cdot)^{-1}$ denote the transpose, complex conjugate, conjugate transpose, and matrix inverse, respectively.
$\mathbb{R}^{a\times b}$ and $\mathbb{C}^{a\times b}$ denote the sets of $a\times b$ real and complex matrices, respectively.
For a vector $\mathbf{x}$, $\|\mathbf{x}\|_2$ denotes its Euclidean norm.
For a matrix $\mathbf{X}$, $\operatorname{Tr}(\mathbf{X})$, $\operatorname{rank}(\mathbf{X})$, and $[\mathbf{X}]_{i,j}$ denote its trace, rank, and $(i,j)$-th element, respectively.
$\operatorname{diag}(x_1,\ldots,x_n)$ denotes a diagonal matrix with diagonal entries $x_1,\ldots,x_n$, and $\mathbf{I}_n$ denotes the $n\times n$ identity matrix.
$\mathbb{E}[\cdot]$ denotes expectation, and $\Re\{\cdot\}$ denotes the real part of a complex number.
$[x]_+\triangleq\max\{x,0\}$ denotes the positive-part operator.
$\nabla_{\mathbf{x}}f$ denotes the Euclidean gradient of $f$ with respect to $\mathbf{x}$.
$\mathcal{CN}(\mu,\sigma^2)$ denotes a circularly symmetric complex Gaussian distribution with mean $\mu$ and variance $\sigma^2$.
$j=\sqrt{-1}$ denotes the imaginary unit, and $\mathcal{O}(\cdot)$ denotes the standard big-$O$ notation.

\section{SYSTEM MODEL}
\label{sec:system_model}

As illustrated in Fig.~\ref{fig场景图}, we consider a UAV downlink communication system consisting of one transmitting UAV and $K$ users. The transmitting UAV is equipped with a uniform planar array (UPA) of directional RAs and simultaneously transmits independent data streams to the users. Each user is equipped with a single isotropic antenna.
The transmitting UAV follows a predetermined trajectory.
Following the commonly used trajectory discretization approach for UAV communications, the continuous flight process is divided into communication slots of duration $T_s>0$. Each slot is chosen sufficiently short such that the relative UAV-user geometry can be approximated as quasi-static, and the UAV and user positions are treated as fixed within the slot~\cite{ZengMobileRelay,ZengUAVTutorial}. In this work, we focus on a given communication slot and jointly design the RA orientations and transmit beamforming.

\subsection{Channel and Signal Model}
\label{subsec:channel_signal_model}
We establish a fixed three-dimensional (3D) Cartesian coordinate system with the $z$-axis pointing vertically upward. The positions of the transmitting UAV and the $k$th user are denoted by
$\mathbf p=[p_x,p_y,p_z]^T\in\mathbb R^3$ and $\mathbf u_k=[u_{k,x},u_{k,y},u_{k,z}]^T\in\mathbb R^3$, where $k\in\mathcal K\triangleq\{1,\ldots,K\}$.
We assume that all users are located below the transmitting UAV (i.e.,
$u_{k,z}<p_z,\ \forall k\in\mathcal K$).
Their positions are treated as fixed within the considered communication slot.

The transmitting UAV is equipped with a UPA consisting of $M=M_xM_y$ directional RAs. The array is parallel to the $x$-$y$ plane and centered at $\mathbf p$. The separation between adjacent RAs is denoted by $d$. The $m$th RA is associated with the centered two-dimensional index $(m_x,m_y)$, where $m_x\in\left\{-\frac{M_x-1}{2},\ldots,\frac{M_x-1}{2}\right\}$ and $m_y\in\left\{-\frac{M_y-1}{2},\ldots,\frac{M_y-1}{2}\right\}$, with $m\in\mathcal M\triangleq\{1,\ldots,M\}$. Accordingly, the position of the $m$th RA relative to the array center is $\mathbf q_m=[m_xd,m_yd,0]^T$, and its absolute position is $\mathbf p_m=\mathbf p+\mathbf q_m$.

We adopt an RA architecture in which the boresight direction of each antenna element can be independently adjusted in three-dimensional space~\cite{ZhengBeixiongRA,WangjiRAIndependent,MatiantianRAIndependent,YeChengzhiRAIndependent}.
For the considered communication slot, all RAs are assumed to start from a common reference orientation. 
Since the users are located below the transmitting UAV, this reference orientation is set to the vertically downward direction, represented by $\mathbf f_0=[0,0,-1]^T$.
Extending the model to consecutive communication slots would generally require the initial orientation of each RA to depend on its terminal orientation in the preceding slot, leading to a state-dependent inter-slot optimization problem, which is beyond the scope of this work.

Let $\theta_m$ denote the angular deviation of the $m$th RA from the reference orientation $\mathbf f_0$, and let $\phi_m$ denote its azimuth angle. Its target boresight direction is represented by the unit vector
\begin{equation}
	\mathbf f_m(\theta_m,\phi_m)
	=
	\begin{bmatrix}
		\sin\theta_m\cos\phi_m\\
		\sin\theta_m\sin\phi_m\\
		-\cos\theta_m
	\end{bmatrix},
	\label{eq:RA_pointing_vector}
\end{equation}
where $0\leq\theta_m\leq\theta_{\max}$ and $-\pi\leq\phi_m\leq\pi$. Here, $\theta_{\max}$ denotes the maximum allowable value of $\theta_m$, with $\theta_{\max}\leq\pi/2$. The orientation variables of all RAs are collected in $\boldsymbol{\vartheta}=[\theta_1,\phi_1,\ldots,\theta_M,\phi_M]^T$.

The distance and unit direction vector from the $m$th RA to the $k$th user are given by
$r_{k,m}=\|\mathbf u_k-\mathbf p_m\|_2$
and
$\mathbf d_{k,m}=(\mathbf u_k-\mathbf p_m)/r_{k,m}$,
respectively.
Let $\epsilon_{k,m}$ denote the angle between the boresight direction of the $m$th RA and $\mathbf d_{k,m}$. Then,
$\cos\epsilon_{k,m}
=\mathbf f_m^T(\theta_m,\phi_m)\mathbf d_{k,m}$.

We adopt a widely used rotationally symmetric directional antenna pattern~\cite{BalanisAntenna}, with the antenna gain modeled as
\begin{equation}
	G(\epsilon)
	=
	\begin{cases}
		G_0\cos^{2p}(\epsilon),
		& 0\leq\epsilon\leq \dfrac{\pi}{2},
		\\[1mm]
		0,
		& \text{otherwise},
	\end{cases}
	\label{eq:antenna_pattern}
\end{equation}
where $p>0$ is the antenna directionality parameter, and
$G_0=2(2p+1)$ is the peak antenna gain chosen to satisfy the power-normalization condition.
As $p$ increases, the main lobe becomes narrower and the radiation pattern becomes more directive.
Accordingly, the antenna gain of the $m$th RA toward the $k$th user is
$G_{k,m}(\boldsymbol{\vartheta})=G(\epsilon_{k,m})$.

Assuming narrowband free-space line-of-sight (LoS) propagation~\cite{XuUAV,BalanisAntenna}, the channel coefficient from the $m$th RA to the $k$th user is
\begin{equation}
	h_{k,m}(\boldsymbol{\vartheta})
	=
	\frac{\lambda}{4\pi r_{k,m}}
	\sqrt{G_{k,m}(\boldsymbol{\vartheta})}
	\exp\left(
	j\frac{2\pi}{\lambda}r_{k,m}
	\right),
	\label{eq:channel_km}
\end{equation}
where $\lambda$ denotes the carrier wavelength.
Define the free-space path gain from
the $m$th RA to the $k$th user as $\beta_{k,m}\triangleq\left(\frac{\lambda}{4\pi r_{k,m}}\right)^2$.
The channel vector from the transmitting UAV to the $k$th user is then given by
\begin{equation}
	\mathbf h_k(\boldsymbol{\vartheta})
	=
	\left[
	h_{k,1}(\boldsymbol{\vartheta}),
	\ldots,
	h_{k,M}(\boldsymbol{\vartheta})
	\right]^T
	\in\mathbb C^{M\times1}.
	\label{eq:channel_vector}
\end{equation}

For the subsequent theoretical analysis, we further assume far-field propagation.
Accordingly, the unit direction vectors and free-space path gains across the transmit array can be approximated as identical for each user~\cite{BalanisAntenna}, i.e.,
\begin{equation}
	\mathbf d_{k,m}\simeq \mathbf d_k, 
	\beta_{k,m}\simeq \beta_k, \forall m\in\mathcal M,
	\label{eq:far_field_approx}
\end{equation}
where $\mathbf d_k
	=
	\frac{\mathbf u_k-\mathbf p}
	{\|\mathbf u_k-\mathbf p\|_2}$
and
$\beta_k
	=
	\left(
	\frac{\lambda}
	{4\pi\|\mathbf u_k-\mathbf p\|_2}
	\right)^2$.
This far-field approximation is used in the subsequent single-user analysis to further simplify the relationship between the RA orientations and the channel gains.

The received signal at the $k$th user is
\begin{equation}
	y_k
	=
	\mathbf h_k^H(\boldsymbol{\vartheta})\mathbf w_k s_k
	+
	\sum_{\substack{j=1\\j\neq k}}^{K}
	\mathbf h_k^H(\boldsymbol{\vartheta})\mathbf w_j s_j
	+
	z_k,
	\label{eq:received_signal}
\end{equation}
where $s_k$ denotes the information symbol intended for the $k$th user and is normalized such that $\mathbb E[|s_k|^2]=1$, $\mathbf w_k\in\mathbb C^{M\times1}$ is the corresponding beamforming vector, and $z_k\sim\mathcal{CN}(0,\sigma_k^2)$ denotes the additive white Gaussian noise (AWGN). The first term in \eqref{eq:received_signal} is the desired signal, while the second term represents multi-user interference.

Define the beamforming matrix as $\mathbf W=[\mathbf w_1,\ldots,\mathbf w_K]\in\mathbb C^{M\times K}$. The transmitted baseband signal is $\mathbf x=\sum_{k=1}^{K}\mathbf w_k s_k$, subject to the total transmit power constraint $\sum_{k=1}^{K}\|\mathbf w_k\|_2^2\leq P_{\max}$.

Accordingly, the SINR at the $k$th user is
\begin{equation}
	\gamma_k(\mathbf W,\boldsymbol{\vartheta})
	=
	\frac{
		\left|
		\mathbf h_k^H(\boldsymbol{\vartheta})\mathbf w_k
		\right|^2
	}{
		\displaystyle
		\sum_{\substack{j=1\\j\neq k}}^{K}
		\left|
		\mathbf h_k^H(\boldsymbol{\vartheta})\mathbf w_j
		\right|^2
		+\sigma_k^2
	},
	\label{eq:SINR}
\end{equation}
where the RA orientations $\boldsymbol{\vartheta}$ affect both the desired signal gain and the multi-user interference through $\mathbf h_k(\boldsymbol{\vartheta})$.

The achievable rate of the link between the transmitting UAV and the $k$th user after the RA reconfiguration is
\begin{equation}
	R_k(\mathbf W,\boldsymbol{\vartheta})
	=
	B\log_2
	\left(
	1+\gamma_k(\mathbf W,\boldsymbol{\vartheta})
	\right),
	\label{eq:instantaneous_rate}
\end{equation}
where $B$ denotes the system bandwidth.
This rate does not account for the time spent on mechanical rotation.

\subsection{Rotation Delay Model and Problem Formulation}
\label{subsec:delay_model}
Since the effective channels vary continuously with the RA orientations
during mechanical rotation, we assume that data transmission is suspended
during RA reconfiguration and resumes once the target orientations are reached.
We model the resulting reconfiguration delay by assuming that all RAs rotate
independently and in parallel at the same constant angular speed $\omega$.
Hence, the reconfiguration time is determined by the largest rotation angle,
i.e., $\max_{m\in\mathcal M}\theta_m/\omega$.

To facilitate the subsequent optimization, we introduce a common rotation budget $\delta$ satisfying $\theta_m\leq\delta$, $\forall m\in\mathcal M$. Accordingly, a duration of $\delta/\omega$ is allocated to RA reconfiguration, and the transmission time fraction is defined as
\begin{equation}
	A(\delta)
	=
	\frac{T_s-\delta/\omega}{T_s}
	=
	1-\frac{\delta}{\omega T_s},
	\label{eq:data_time_fraction}
\end{equation}
where $\delta/\omega\leq T_s$ is required for feasibility.

Accounting for the rotation time, the effective transmission rate of the link between the transmitting UAV and the $k$th user is defined as
\begin{equation}
	\bar{R}_k(\mathbf W,\boldsymbol{\vartheta},\delta)
	=
	A(\delta)
	B\log_2
	\left(
	1+\gamma_k(\mathbf W,\boldsymbol{\vartheta})
	\right).
	\label{eq:effective_rate}
\end{equation}

Increasing $\delta$ enlarges the feasible set of RA orientations, but reduces the time available for data transmission.

We aim to maximize the minimum effective transmission rate among all users by jointly optimizing the transmit beamforming matrix $\mathbf W$, the RA orientations $\boldsymbol{\vartheta}$, and the common rotation budget $\delta$.
From \eqref{eq:effective_rate}, since all users share the same transmission time fraction $A(\delta)$ and $\log_2(1+x)$ is monotonically increasing in $x$, we have
\[
	\min_{k\in\mathcal K}\bar{R}_k
	=
	A(\delta)B
	\log_2\left(
	1+\min_{k\in\mathcal K}
	\gamma_k(\mathbf W,\boldsymbol{\vartheta})
	\right).
\]
The joint optimization problem is therefore formulated as
\refstepcounter{prob}
\label{prob:P1}
\begin{subequations}
	\begin{align}
		(\theprob):\quad
		\max_{\mathbf W,\boldsymbol{\vartheta},\delta}\quad
		&
		A(\delta)B
		\log_2\left(
		1+\min_{k\in\mathcal K}
		\gamma_k(\mathbf W,\boldsymbol{\vartheta})
		\right)
		\label{prob:P1_obj}
		\\
		\mathrm{s.t.}\quad
		\mathrm{C1}:~&
		\sum_{k=1}^{K}\|\mathbf w_k\|_2^2
		\leq P_{\max},
		\label{prob:P1_C1}
		\\
		\mathrm{C2}:~&
		0\leq\theta_m\leq\delta,
		\qquad \forall m\in\mathcal M,
		\label{prob:P1_C2}
		\\
		\mathrm{C3}:~&
		-\pi\leq\phi_m\leq\pi,
		\qquad \forall m\in\mathcal M,
		\label{prob:P1_C3}
		\\
		\mathrm{C4}:~&
		0\leq\delta\leq\delta_{\max},
		\label{prob:P1_C4}
	\end{align}
\end{subequations}
where $\delta_{\max}\triangleq\min\{\theta_{\max},\omega T_s\}$.
Constraint $\mathrm{C1}$ imposes the total transmit power limit.
Constraints $\mathrm{C2}$ and $\mathrm{C3}$ specify the feasible RA orientations.
Constraint $\mathrm{C4}$ accounts for both the maximum mechanical rotation angle and the requirement $\delta/\omega\leq T_s$.

Since
$A(\delta)=1-\delta/(\omega T_s)$
is monotonically decreasing in $\delta$, for any given
$\mathbf W$ and $\boldsymbol{\vartheta}$, the objective is maximized by the smallest feasible $\delta$ satisfying $\mathrm{C2}$.
Therefore, at the optimum of problem~\eqref{prob:P1},
\begin{equation}
	\delta^\star
	=
	\max_{m\in\mathcal M}\theta_m^\star.
	\label{eq:optimal_delta}
\end{equation}

Thus, the rotation budget constraint is tight at the optimum, and $\delta^\star$ coincides with the actual maximum rotation angle.

Problem~\eqref{prob:P1} is nonconvex due to the coupling among the beamforming matrix $\mathbf W$, the RA orientations $\boldsymbol{\vartheta}$, and the rotation budget $\delta$. In particular, $\boldsymbol{\vartheta}$ and $\mathbf W$ jointly determine the multi-user SINRs, while $\delta$ affects both the feasible RA orientations and the transmission time fraction $A(\delta)$. Hence, problem~\eqref{prob:P1} is difficult to solve directly. In the following, we exploit its structure for the single-user and multi-user systems separately.
\section{Single-User System}

We first consider the single-user system (i.e., $K=1$) to characterize the tradeoff between the gain from RA rotation and the mechanical rotation delay in problem~$(\mathrm{P1})$.
For notational simplicity, the user index is omitted in the single-user system, e.g., $r_m$, $\mathbf d_m$, and $\beta_m$ are used in place of $r_{1,m}$, $\mathbf d_{1,m}$, and $\beta_{1,m}$, respectively.
Problem \eqref{prob:P1} then reduces to
\refstepcounter{prob}
\label{prob:P2}
\begin{subequations}
	\begin{align}
		(\theprob):\quad
		\max_{\mathbf w,\boldsymbol{\vartheta},\delta}\quad
		&
		A(\delta)B
		\log_2
		\left(
		1+
		\frac{
			\left|
			\mathbf h^H(\boldsymbol{\vartheta})\mathbf w
			\right|^2
		}{
			\sigma^2
		}
		\right)
		\label{prob:P2_obj}
		\\
		\mathrm{s.t.}\quad
		&
		\|\mathbf w\|_2^2\leq P_{\max},
		\label{prob:P2_power}
		\\
		&
		0\leq\theta_m\leq\delta,
		\quad \forall m\in\mathcal M,
		\label{prob:P2_theta}
		\\
		&
		-\pi\leq\phi_m\leq\pi,
		\quad \forall m\in\mathcal M,
		\label{prob:P2_phi}
		\\
		&
		0\leq\delta\leq\delta_{\max}.
		\label{prob:P2_delta}
	\end{align}
\end{subequations}
\subsection{Problem Reformulation}
\label{subsec:single_reformulation}

For any given RA orientations $\boldsymbol{\vartheta}$, the channel vector $\mathbf h(\boldsymbol{\vartheta})$ is fixed. By the Cauchy--Schwarz inequality,
$\left|
\mathbf h^H(\boldsymbol{\vartheta})\mathbf w
\right|^2
\leq
\|\mathbf h(\boldsymbol{\vartheta})\|_2^2
\|\mathbf w\|_2^2
\leq
P_{\max}
\|\mathbf h(\boldsymbol{\vartheta})\|_2^2$.
Hence, when $\mathbf h(\boldsymbol{\vartheta})\neq\mathbf 0$, the optimal transmit beamforming vector in problem~$(\mathrm{P2})$ is given by the maximum-ratio transmission (MRT) beamformer,
\begin{equation}
	\mathbf w^{\star}
	=
	\sqrt{P_{\max}}
	\frac{
		\mathbf h(\boldsymbol{\vartheta})
	}{
		\|\mathbf h(\boldsymbol{\vartheta})\|_2
	}.
	\label{eq:single_mrt}
\end{equation}

Substituting \eqref{eq:single_mrt} into \eqref{prob:P2_obj}, the maximum signal-to-noise ratio (SNR) for given RA orientations is
\begin{equation}
	\gamma(\boldsymbol{\vartheta})
	=
	\frac{P_{\max}}{\sigma^2}
	\|\mathbf h(\boldsymbol{\vartheta})\|_2^2.
	\label{eq:single_snr_pointing}
\end{equation}

Combining \eqref{eq:antenna_pattern}, \eqref{eq:channel_km}, and \eqref{eq:single_snr_pointing}, the instantaneous SNR can be expressed as
\begin{equation}
	\gamma(\boldsymbol{\vartheta})
	=
	\frac{P_{\max}G_0}{\sigma^2}
	\sum_{m=1}^{M}
	\beta_m
	\left[
		\mathbf f_m^T(\theta_m,\phi_m)\mathbf d_m
	\right]_+^{2p},
	\label{eq:single_snr_sum}
\end{equation}
where
$\beta_m\triangleq
\left(\frac{\lambda}{4\pi r_m}\right)^2$
denotes the free-space path gain of the $m$th RA.
For a given common rotation budget $\delta$, the gain of each RA in \eqref{eq:single_snr_sum} can be optimized independently.

To characterize the resulting optimal orientation, let $\alpha_m$ denote the angle between $\mathbf f_0$ and $\mathbf d_m$, defined as
\begin{equation}
	\alpha_m
	\triangleq
	\arccos
	\left(
		\mathbf f_0^T\mathbf d_m
	\right),
	\qquad
	0\leq\alpha_m<\frac{\pi}{2},
	\label{eq:single_alpha_m}
\end{equation}
which is referred to as the initial off-boresight angle of the $m$th RA.

For a given rotation budget $\delta$, $A(\delta)$ is constant. Therefore, maximizing the effective transmission rate is equivalent to maximizing the instantaneous SNR in \eqref{eq:single_snr_sum}. The SNR in \eqref{eq:single_snr_sum} is a sum of gain terms associated with individual RAs, while the orientation constraints are independent across the RAs. Hence, the orientation of each RA can be optimized separately. The resulting optimization for each RA has the same structure as the single-user RA pointing problem in \cite{ZhengBeixiongRA}. Accordingly, the optimal azimuth and deflection angles of the $m$th RA for a given $\delta$ are given by
\begin{equation}
	\phi_m^\star=
	\operatorname{atan2}(d_{m,y},d_{m,x}),
    \label{eq:single_opt_phi}
\end{equation}

\begin{equation}
	\theta_m^\star(\delta)=
	\min\{\alpha_m,\delta\}.
    \label{eq:single_opt_theta}
\end{equation}

When $\theta_m^\star(\delta)=0$, the azimuth angle is immaterial and can be chosen arbitrarily. For a given $\delta$, the residual off-boresight angle under the SNR-maximizing orientation is
\begin{equation}
	\varepsilon_m^\star(\delta)
	=
	[\alpha_m-\delta]_+.
	\label{eq:single_residual_misalignment}
\end{equation}

Substituting these orientations into \eqref{eq:single_snr_sum} yields the maximum instantaneous SNR for a given $\delta$ as

\begin{equation}
	\gamma^\star(\delta)
	=
	\frac{P_{\max}G_0}{\sigma^2}
	\sum_{m=1}^{M}
	\beta_m
	\cos^{2p}
	\left(
		\varepsilon_m^\star(\delta)
	\right).
	\label{eq:single_opt_snr_delta}
\end{equation}

From \eqref{eq:single_opt_theta}, when $\delta\geq\max_{m\in\mathcal M}\alpha_m$, all RAs are fully aligned with the user. Further increasing $\delta$ does not improve the maximum SNR, but it reduces the time available for data transmission. Therefore, it suffices to optimize the rotation budget over
\begin{equation}
	0\leq\delta\leq\bar{\delta},
	\qquad
	\bar{\delta}
	\triangleq
	\min
	\left\{
		\delta_{\max},
		\max_{m\in\mathcal M}\alpha_m
	\right\}.
	\label{eq:single_delta_bar}
\end{equation}

For any given $\delta$ in this interval, the optimal transmit beamformer and RA orientations are already determined by the preceding results. Hence, problem~$(\mathrm{P2})$ can be equivalently reformulated as
\refstepcounter{prob}
\label{prob:P3}
\begin{equation}
    (\theprob):\quad
    \max_{0\leq\delta\leq\bar{\delta}}
    \;
    A(\delta)B
    \log_2\left(1+\gamma^\star(\delta)\right),
    \label{eq:P3_obj}
\end{equation}
where $\gamma^\star(\delta)$ is given in
\eqref{eq:single_opt_snr_delta}.
Thus, the original joint optimization over the transmit beamforming vector, the $M$ RA orientations, and the rotation budget is reduced to a one-dimensional optimization over $\delta$.

\subsection{Optimal Rotation Design Under the Far-Field Approximation}

Applying the far-field approximation in \eqref{eq:far_field_approx} and noting that all RAs share the same initial boresight direction $\mathbf f_0$, we obtain
\begin{equation}
	\alpha_m
	=
	\arccos(\mathbf f_0^T\mathbf d_m)
	\simeq
	\alpha,
	\qquad
	\beta_m\simeq\beta,
	\quad m\in\mathcal M.
	\label{eq:single_far_field_parameters}
\end{equation}

Accordingly, the upper bound in \eqref{eq:single_delta_bar} reduces to
$\bar{\delta}=\min\{\delta_{\max},\alpha\}$.
From \eqref{eq:single_residual_misalignment}, the residual off-boresight angle of each RA reduces to $[\alpha-\delta]_+$. Since $0\leq\delta\leq\bar{\delta}\leq\alpha$, we have $[\alpha-\delta]_+=\alpha-\delta$. Hence, \eqref{eq:single_opt_snr_delta} can be simplified as
\begin{equation}
	\gamma^\star_{\mathrm{FF}}(\delta)
	=
	C\cos^{2p}(\alpha-\delta),
	\qquad
	0\leq\delta\leq\bar{\delta},
	\label{eq:far_field_snr}
\end{equation}
where $C\triangleq P_{\max}G_0M\beta/\sigma^2$ denotes the peak-aligned SNR, corresponding to perfect alignment of all RA boresights with the user under maximum-ratio transmission.

Substituting \eqref{eq:far_field_snr} into problem~$(\mathrm{P3})$ yields the following rotation optimization problem under the far-field approximation:
\refstepcounter{prob}
\label{prob:far_field}
\begin{equation}
    (\theprob):\quad
    \max_{0\leq\delta\leq\bar{\delta}}
    \bar{R}_{\mathrm{FF}}(\delta),
    \label{eq:far_field_problem}
\end{equation}
where
\begin{equation}
	\bar{R}_{\mathrm{FF}}(\delta)
	=
	\left(
	1-\frac{\delta}{\omega T_s}
	\right)
	B
	\log_2
	\left[
	1+C\cos^{2p}(\alpha-\delta)
	\right].
	\label{eq:far_field_rate}
\end{equation} 

Increasing the rotation budget has two opposing effects on the effective transmission rate. A larger $\delta$ reduces the residual off-boresight angle and thereby improves the instantaneous SNR after rotation. However, it also increases the mechanical rotation delay and  leaves less time for data transmission.

To further analyze problem~(\ref{prob:far_field}), define
	$q(\delta)
	\triangleq
	\cos^{2p}(\alpha-\delta)$.
The first-order derivative of \eqref{eq:far_field_rate} with respect to $\delta$ is given by
\begin{align}
	\bar{R}_{\mathrm{FF}}'(\delta)
	={}&
	-\frac{B}{\omega T_s\ln 2}
	\ln\left[1+Cq(\delta)\right]
	\nonumber\\
	&
	+
	\frac{B(\omega T_s-\delta)}{\omega T_s\ln 2}
	\frac{
		2Cp\,q(\delta)\tan(\alpha-\delta)
	}{
		1+Cq(\delta)
	}.
	\label{eq:far_field_rate_derivative}
\end{align}
The first term in \eqref{eq:far_field_rate_derivative} captures the rate loss caused by the reduction in data transmission time, whereas the second term captures the rate gain due to the SNR improvement resulting from the reduced off-boresight angle. Therefore, if the optimal solution lies in the interior of the feasible interval, it must satisfy
\begin{equation}
	(\omega T_s-\delta)
	\frac{
		2Cp\,q(\delta)\tan(\alpha-\delta)
	}{
		1+Cq(\delta)
	}
	=
	\ln\!\left[1+Cq(\delta)\right].
	\label{eq:marginal_balance}
\end{equation}
This condition indicates that, at an interior optimum, the rate gain resulting from the reduced off-boresight angle is exactly offset by the rate loss caused by the additional rotation delay.

We next characterize the globally optimal rotation design for
problem~(\ref{prob:far_field}).

\begin{theorem}
\label{thm:single_optimal_rotation}
For $p>0$, consider the common rotation budget $\delta\in[0,\bar{\delta}]$
and the common initial off-boresight angle $\alpha$ under the far-field approximation.
If
\begin{equation}
	1+
	C\cos^{2p}(\alpha-\delta)
	\geq
	2p\sin^2(\alpha-\delta),
	\qquad
	0\leq\delta\leq\bar{\delta},
	\label{eq:concavity_condition}
\end{equation}
then $\bar{R}_{\mathrm{FF}}(\delta)$ is strictly concave over
$[0,\bar{\delta}]$.
Consequently, problem~(\ref{prob:far_field}) has a unique global optimum given by
\begin{equation}
	\delta^\star
	=
	\begin{cases}
		0,
		&\bar{R}_{\mathrm{FF}}'(0)\leq 0,
		\\[1mm]
		\bar{\delta},
		& \bar{R}_{\mathrm{FF}}'(\bar{\delta})\geq 0,
		\\[1mm]
		\delta_0,
		& \text{otherwise},
	\end{cases}
	\label{eq:optimal_delta_exact}
\end{equation}
where $\delta_0\in(0,\bar{\delta})$ is the unique solution to
\eqref{eq:marginal_balance}.
\end{theorem}

\begin{IEEEproof}
See Appendix~\ref{app:proof_theorem1}. 
\end{IEEEproof}

Theorem~\ref{thm:single_optimal_rotation} further reveals three optimal rotation regimes determined by the tradeoff between the rate gain from further rotation and the rate loss caused by the additional rotation delay. If the rate gain from further rotation at $\delta=0$ does not exceed the corresponding rate loss, increasing the rotation budget cannot improve the effective transmission rate, and no rotation is optimal. Conversely, if the rate gain at the upper bound $\bar{\delta}$ is no smaller than the rate loss, the effective transmission rate increases throughout the feasible interval, and the optimal rotation budget is $\bar{\delta}$. 
Otherwise, the optimal solution lies in the interior of the feasible interval, corresponding to partial rotation, where the unique solution $\delta_0$ satisfies \eqref{eq:marginal_balance}. Since $\bar{R}_{\mathrm{FF}}(\delta)$ is strictly concave, this interior solution is the unique global optimum and can be efficiently obtained via a one-dimensional bisection search.

\subsection{Closed-Form Solution} 
\label{subsec:closed_form}
Theorem~\ref{thm:single_optimal_rotation} establishes the global solution structure of problem \eqref{prob:far_field} under the sufficient condition in \eqref{eq:concavity_condition}.
When the optimum lies in the interior of the feasible interval, however, the optimal rotation budget is determined by the transcendental equation~\eqref{eq:marginal_balance}, which generally admits no closed-form solution. We therefore derive a closed-form solution using a second-order expansion with respect to the residual off-boresight angle.

For convenience, define
$x\triangleq\alpha-\delta$ and
$a\triangleq\omega T_s-\alpha$,
where $x$ denotes the residual off-boresight angle after RA rotation.
Substituting $\delta=\alpha-x$ into \eqref{eq:far_field_rate}, we obtain
\begin{equation}
	\mathcal R(x)
	\triangleq
	\bar{R}_{\mathrm{FF}}(\alpha-x)
	=
	\frac{B}{\omega T_s\ln2}
	(a+x)
	\ln\!\left(
	1+C\cos^{2p}x
	\right),
	\label{eq:rate_x}
\end{equation}
where
$\alpha-\bar{\delta}\leq x\leq\alpha$.

To obtain a closed-form solution, we apply a second-order Taylor expansion to the logarithmic term around $x=0$. The resulting approximate objective has the following property.

\begin{theorem}
	\label{thm:closed_form_rotation}
	If $a=\omega T_s-\alpha>0$, the resulting approximate objective is strictly concave for $x\geq0$ and has a unique
	nonnegative stationary point given by
	\begin{equation}
		x_{\mathrm{cf}}
		=
		\frac{
			\sqrt{
				a^2+
				\frac{3(1+C)\ln(1+C)}{Cp}
			}
			-a
		}{3}.
		\label{eq:x_cf}
	\end{equation}
	Accordingly, the closed-form solution for the rotation budget is
	\begin{equation}
		\delta_{\mathrm{cf}}
		=
		\Pi_{[0,\bar{\delta}]}
		\left(
		\alpha-x_{\mathrm{cf}}
		\right),
		\label{eq:delta_cf}
	\end{equation}
	where $\Pi_{[0,\bar{\delta}]}(\cdot)$ denotes the projection onto
	$[0,\bar{\delta}]$.
\end{theorem}

\begin{IEEEproof}
    See Appendix~\ref{app:proof_theorem2}.
\end{IEEEproof}

Theorem~\ref{thm:closed_form_rotation} yields a closed-form rotation budget without numerically solving \eqref{eq:marginal_balance}. The accuracy and effective-rate performance of the closed-form solution are evaluated numerically in Section \ref{SimulationResults}.

\section{Multi-User System}
\label{sec:multiuser}

We next consider the general multi-user system with $K>1$.
Unlike the single-user system, the orientation of each RA affects the channels of multiple users and is coupled with the transmit beamforming design.
Therefore, for a given rotation budget, the RA orientations can no longer be optimized independently.

For a given rotation budget $\delta$, define the maximum achievable minimum SINR as
\refstepcounter{prob}
\label{prob:fixed_delta}
\begin{subequations}
	\begin{align}
	(\theprob):\quad
	g(\delta)
	\triangleq
	\max_{\mathbf W,\boldsymbol{\vartheta}}
	\quad&
	\min_{k\in\mathcal K}
	\gamma_k(\mathbf W,\boldsymbol{\vartheta})
	\label{prob:fixed_delta_obj}
	\\
	\mathrm{s.t.}\quad
	&
	\sum_{k=1}^{K}\|\mathbf w_k\|_2^2
	\le P_{\max},
	\label{prob:fixed_delta_power}
	\\
	&
	0\le\theta_m\le\delta,
	\qquad \forall m\in\mathcal M,
	\label{prob:fixed_delta_theta}
	\\
	&
	-\pi\le\phi_m\le\pi,
	\qquad \forall m\in\mathcal M.
	\label{prob:fixed_delta_phi}
	\end{align}
\end{subequations}

With this definition, problem~\eqref{prob:P1} can be equivalently reformulated as
\refstepcounter{prob}
\label{prob:outer_delta}
\begin{equation}
	(\theprob):\quad
	\max_{0\le\delta\le\delta_{\max}}
	A(\delta)B
	\log_2\!\left(1+g(\delta)\right).
	\label{eq:multi_outer_problem}
\end{equation}

This reformulation separates the optimization into an inner communication design for a fixed $\delta$ and an outer optimization over the rotation budget.
For $\delta_1\leq\delta_2$, the feasible sets of
problem~\eqref{prob:fixed_delta} satisfy
$\mathcal F(\delta_1)\subseteq\mathcal F(\delta_2)$.
Therefore,
$g(\delta_1)\leq g(\delta_2)$,
and hence $g(\delta)$ is nondecreasing in $\delta$.

However, evaluating $g(\delta)$ for a given $\delta$ still requires the
joint optimization of transmit beamforming and RA orientations, which is
difficult due to their strong coupling.
To obtain a tractable solution, we restrict the transmit beamforming to
a zero-forcing (ZF) structure and denote the corresponding optimal value by
$g_{\rm ZF}(\delta)$.
Since ZF imposes an additional constraint on the transmit beamforming,
we have
\begin{equation}
	g_{\rm ZF}(\delta)\leq g(\delta)
\end{equation}
whenever ZF is feasible.
We next exploit the ZF structure to eliminate the beamforming and power
allocation variables, thereby reducing the fixed-budget problem to an
optimization over the RA orientations.

\subsection{ZF-Based RA Pointing Optimization}
\label{subsec:zf_pointing}

For a given rotation budget $\delta$, let
$\mathbf F=[\mathbf f_1,\ldots,\mathbf f_M]$
denote the RA orientation matrix,
$\mathbf H(\mathbf F)
=[\mathbf h_1(\mathbf F),\ldots,\mathbf h_K(\mathbf F)]$
denote the multi-user channel matrix, and
$\mathbf G(\mathbf F)=\mathbf H^H(\mathbf F)\mathbf H(\mathbf F)$
denote the corresponding Gram matrix.

When $M\geq K$ and $\operatorname{rank}(\mathbf H)=K$,
ZF beamforming can eliminate multi-user interference.
Define the unnormalized ZF beamforming matrix as
$\widetilde{\mathbf W}
=\mathbf H\mathbf G^{-1}
=[\widetilde{\mathbf w}_1,\ldots,\widetilde{\mathbf w}_K]$,
and let
$\mathbf w_k=\sqrt{q_k}\widetilde{\mathbf w}_k$,
where $q_k\geq0$ is the power coefficient.
Since
$\mathbf H^H\widetilde{\mathbf W}=\mathbf I_K$,
the SINR of user $k$ is $q_k/\sigma_k^2$, while the total transmit power is
$\sum_{k=1}^{K}q_k[\mathbf G^{-1}]_{k,k}$.
Under the max-min fairness criterion, the optimal power allocation equalizes
the ZF SINRs, i.e., $q_k^\star=\eta_{\rm ZF}\sigma_k^2$.
Defining
$\mathbf D=\operatorname{diag}(\sigma_1^2,\ldots,\sigma_K^2)$,
the maximum common ZF SINR is
\begin{equation}
	\eta_{\rm ZF}(\mathbf F)
	=
	\frac{P_{\max}}
	{\operatorname{Tr}
	\left[
	\mathbf D\mathbf G^{-1}(\mathbf F)
	\right]} .
	\label{eq:zf_common_sinr}
\end{equation}

It follows from \eqref{eq:zf_common_sinr} that, for a given $\delta$,
maximizing the common ZF SINR is equivalent to minimizing
$\operatorname{Tr}[\mathbf D\mathbf G^{-1}(\mathbf F)]$.
Therefore, for a given $\delta$, the inner problem~\eqref{prob:fixed_delta} under ZF beamforming can be reduced to the following RA pointing optimization:
\refstepcounter{prob}
\label{prob:zf_pointing}
\begin{subequations}
	\begin{align}
	(\theprob):\quad
	\min_{\mathbf F}\quad&
	J(\mathbf F)
	\triangleq
	\operatorname{Tr}
	\left[
		\mathbf D\mathbf G^{-1}(\mathbf F)
	\right]
	\label{prob:zf_pointing_obj}\\
	\mathrm{s.t.}\quad&
	\|\mathbf f_m\|_2=1,\qquad
	\mathbf f_0^T\mathbf f_m\geq\cos\delta,
	\quad \forall m\in\mathcal M .
	\label{prob:zf_pointing_cons}
	\end{align}
\end{subequations}

\begin{figure*}[!b]
\normalsize
\hrulefill
\begin{equation}
\nabla_{\mathbf f_m}J
=
-2p\Re
\left\{
\sum_{\substack{k\in\mathcal K:\\
\mathbf f_m^T\mathbf d_{k,m}>0}}
\left[
\mathbf H\mathbf G^{-1}
\mathbf D\mathbf G^{-1}
\right]_{m,k}^{*}
c_{k,m}
\left(
\mathbf f_m^T\mathbf d_{k,m}
\right)^{p-1}
\mathbf d_{k,m}
\right\}.
\label{eq:analytic_gradient}
\end{equation}
\end{figure*}

Let $J_{\rm ZF}^{\star}(\delta)$ denote the optimal objective value of
problem~\eqref{prob:zf_pointing} for a given $\delta$.
Then,
$g_{\rm ZF}(\delta)
=
P_{\max}/J_{\rm ZF}^{\star}(\delta)$.
When all users have identical noise variances, i.e.,
$\sigma_k^2=\sigma^2,\ \forall k\in\mathcal K$,
we have
$J(\mathbf F)
=
\sigma^2
\sum_{k=1}^{K}
\lambda_k^{-1}(\mathbf G)$,
where $\lambda_k(\mathbf G)$ denotes the $k$th eigenvalue of $\mathbf G$.
Thus, small eigenvalues of $\mathbf G$ incur a large ZF power penalty,
indicating that RA pointing should account for both the channel strength
and the spatial separability of the user channels.

We next derive the analytical gradient used to iteratively optimize problem~\eqref{prob:zf_pointing} for $p\geq1$.
From \eqref{eq:channel_km}, the channel coefficient between the $m$th RA
and the $k$th user can be written as
\begin{equation}
	h_{k,m}(\mathbf f_m)
	=
	c_{k,m}
	\left[
	\mathbf f_m^T\mathbf d_{k,m}
	\right]_+^p,
	\label{eq:channel_compact_mu}
\end{equation}
where
$c_{k,m}\triangleq
\sqrt{\beta_{k,m}G_0}
\exp(j2\pi r_{k,m}/\lambda)$
is independent of $\mathbf f_m$.

\begin{proposition}
\label{prop:analytic_gradient}
Suppose that $\mathbf H(\mathbf F)$ has full column rank.
For $p>1$, the Euclidean gradient of $J(\mathbf F)$ with respect to
the RA orientation vector $\mathbf f_m$ is given by
\eqref{eq:analytic_gradient}.
For $p=1$, the same expression is adopted under the generalized-gradient
convention specified in the proof.
\end{proposition}

\begin{IEEEproof}
For $p>1$, $[x]_+^p$ is continuously differentiable.
For $p=1$, \eqref{eq:analytic_gradient} gives the ordinary gradient
whenever $\mathbf f_m^T\mathbf d_{k,m}\neq0$.
At $\mathbf f_m^T\mathbf d_{k,m}=0$, we select the zero element of the
Clarke generalized gradient of $[x]_+$ at $x=0$.

From
$J=\operatorname{Tr}(\mathbf D\mathbf G^{-1})$ and
$\mathrm d\mathbf G^{-1}
=-\mathbf G^{-1}(\mathrm d\mathbf G)\mathbf G^{-1}$,
we obtain
\begin{equation}
\mathrm dJ
=
-\operatorname{Tr}
\!\left(
\mathbf G^{-1}\mathbf D\mathbf G^{-1}
\mathrm d\mathbf G
\right),
\end{equation}
where
$\mathrm d\mathbf G
=(\mathrm d\mathbf H)^H\mathbf H
+\mathbf H^H\mathrm d\mathbf H$.
Combining this result with the differential of
\eqref{eq:channel_compact_mu} with respect to $\mathbf f_m$ and
collecting the corresponding terms yields
\eqref{eq:analytic_gradient}.
\end{IEEEproof}

Since $\mathbf f_m$ lies on the unit sphere, we first project
\eqref{eq:analytic_gradient} onto the tangent space at $\mathbf f_m$
and generate a candidate point along the negative gradient direction:
\begin{equation}
	\mathbf g_m
	=
	(\mathbf I_3-\mathbf f_m\mathbf f_m^T)
	\nabla_{\mathbf f_m}J,
	\qquad
	\widehat{\mathbf f}_m
	=
	\frac{\mathbf f_m-\mu\mathbf g_m}
	{\|\mathbf f_m-\mu\mathbf g_m\|_2},
	\label{eq:tangent_update}
\end{equation}
where $\mu>0$ denotes the step size.

If
$\mathbf f_0^T\widehat{\mathbf f}_m\geq\cos\delta$,
then $\widehat{\mathbf f}_m$ already satisfies the spherical-cap constraint.
Otherwise, it is projected onto the boundary of the spherical cap.
The resulting update is
\begin{equation}
	\mathbf f_m^{+}
	=
	\begin{cases}
		\widehat{\mathbf f}_m,
		&
		\mathbf f_0^T\widehat{\mathbf f}_m
		\geq\cos\delta,
		\\[1mm]
		\displaystyle
		\cos\delta\,\mathbf f_0
		+
		\sin\delta\,
		\frac{
		(\mathbf I_3-\mathbf f_0\mathbf f_0^T)
		\widehat{\mathbf f}_m
		}{
		\|
		(\mathbf I_3-\mathbf f_0\mathbf f_0^T)
		\widehat{\mathbf f}_m
		\|_2
		},
		&
		\text{otherwise}.
	\end{cases}
	\label{eq:spherical_cap_update}
\end{equation}
The step size $\mu$ is selected by backtracking line search to ensure that
$J(\mathbf F)$ does not increase after each update.
The AG iterations terminate when the relative objective change is below $\epsilon_{\rm AG}$ or the maximum number of iterations $I_{\rm AG}$ is reached. 

\subsection{Proposed AG-ZF Algorithm}
\label{subsec:agzf}
Based on the analytical gradient and feasible update derived above, we propose the AG-ZF algorithm. 
For each candidate rotation budget $\delta$, the fixed-budget RA pointing problem is solved using analytical-gradient updates with backtracking line search.
Since the optimized orientations may not fully use the prescribed rotation budget, the effective transmission rate is evaluated using the actual maximum rotation angle
$a(\mathbf F)=\max_{m\in\mathcal M}\arccos(\mathbf f_0^T\mathbf f_m)$.
The outer search starts from $Q_0$ rotation-budget samples over the feasible interval. At each refinement stage, new candidate budgets are generated around the actual rotation angle of the current best solution.
After $S$ refinement stages, the solution with the highest effective transmission rate is selected.
The overall procedure is summarized in Algorithm~\ref{alg:agzf}.

\begin{algorithm}[!t]
    \caption{Proposed AG-ZF Algorithm}
    \label{alg:agzf}
    \begin{algorithmic}[1]
        \REQUIRE
        $Q_0$, $S$, $I_{\rm AG}$, $\epsilon_{\rm AG}$.
        \ENSURE
        $\widehat{\delta}$, $\widehat{\mathbf F}$,
        $\widehat{\mathbf W}$.

        \STATE Set
        $\mathcal D_\delta^{(0)}
        =
        \left\{
        q\delta_{\max}/(Q_0-1)
        \right\}_{q=0}^{Q_0-1}$.

        \STATE Set $R_{\rm best}=-\infty$.

        \FOR{$s=0,\ldots,S$}
            \FOR{$\delta\in\mathcal D_\delta^{(s)}$}

                \STATE Initialize a feasible
                $\mathbf F^{(0)}$ and set $i=0$.

                \REPEAT
                    \STATE Compute
                    $\nabla_{\mathbf f_m}J$ from
                    \eqref{eq:analytic_gradient}.

                    \STATE Select $\mu$ by backtracking
                    and update $\mathbf f_m$ using
                    \eqref{eq:tangent_update} and
                    \eqref{eq:spherical_cap_update}.

                    \STATE $i\leftarrow i+1$.

                \UNTIL{
                    $\frac{|J^{(i)}-J^{(i-1)}|}
                    {J^{(i-1)}}
                    \leq\epsilon_{\rm AG}$
                    or $i=I_{\rm AG}$}

                \STATE Set
                $\widehat{\mathbf F}(\delta)
                =\mathbf F^{(i)}$.

                \STATE Set
                $\delta_a(\delta)
                =a(\widehat{\mathbf F}(\delta))$
                and
                $R(\delta)
                =A(\delta_a)B
                \log_2\!\left(
                1+\eta_{\rm ZF}
                (\widehat{\mathbf F}(\delta))
                \right)$.

                \IF{$R(\delta)>R_{\rm best}$}
                    \STATE Set
                    $R_{\rm best}\leftarrow R(\delta)$
                    and
                    $\mathbf F_{\rm best}
                    \leftarrow\widehat{\mathbf F}(\delta)$.
                \ENDIF

            \ENDFOR

            \STATE Select the solution with
            the largest $R(\delta)$ in the current stage,
            denoted by $\widehat{\mathbf F}_s$.

            \IF{$s<S$}
                \STATE Let $\delta_{\rm L}$
                and $\delta_{\rm R}$ be the
                adjacent evaluated points to
                $a(\widehat{\mathbf F}_s)$.

                \STATE Set
                $\mathcal D_\delta^{(s+1)}
                =
                \left\{
                \frac{
                \delta_{\rm L}
                +a(\widehat{\mathbf F}_s)
                }{2},
                \frac{
                a(\widehat{\mathbf F}_s)
                +\delta_{\rm R}
                }{2}
                \right\}$.
            \ENDIF

        \ENDFOR

        \STATE Set
        $\widehat{\mathbf F}
        =\mathbf F_{\rm best}$
        and
        $\widehat{\delta}
        =a(\widehat{\mathbf F})$.

        \STATE Set
        $\widetilde{\mathbf W}
        =\mathbf H(\widehat{\mathbf F})
        \mathbf G^{-1}(\widehat{\mathbf F})$
        and
        $\widehat{\mathbf w}_k
        =\sqrt{\eta_{\rm ZF}(\widehat{\mathbf F})\sigma_k^2}\,
        \widetilde{\mathbf w}_k$.

        \RETURN
        $\widehat{\delta},
        \widehat{\mathbf F},
        \widehat{\mathbf W}$.

    \end{algorithmic}
\end{algorithm}

For any fixed $\delta$, the backtracking line search ensures that $J(\mathbf F)$ does not increase after each accepted update.
Since $J(\mathbf F)>0$, the resulting objective-value sequence is monotonically nonincreasing and bounded below, and therefore converges to a finite value.

For complexity comparison, we consider the delay-aware alternating optimization (DA-AO) scheme described in Section~V.
Let $Q_{\rm tot}$ denote the total number of fixed-budget problems evaluated during the outer rotation search.
For each rotation budget, DA-AO alternately solves SOCP-based approximations of the beamforming and RA pointing subproblems, whose numbers of real-valued optimization variables are $2MK+K+1$ and $2M+K+1$, respectively.
Using a generic interior-point method, its overall computational complexity can be conservatively estimated as
$\widetilde{\mathcal O}\!\Big(
Q_{\rm tot}I_{\rm AO}\big[
N_{\eta}I_{\rm IP,W}(2MK+K+1)^3
+I_{\rm IP,F}(2M+K+1)^3
+MK^2
\big]
\Big),$
where $I_{\rm AO}$ denotes the number of alternating-optimization iterations,
$I_{\rm IP,W}$ and $I_{\rm IP,F}$ denote the numbers of interior-point iterations required by the beamforming and RA-pointing subproblems, respectively, and $N_\eta$ denotes the average number of iterations used for the max-min SINR search in the beamforming subproblem.
The delay-unaware alternating optimization (DU-AO) scheme has the same per-iteration complexity, but does not require the outer rotation-budget search.

In contrast, AG-ZF avoids solving generic convex optimization subproblems.
Its main computational cost arises from the matrix operations required for evaluating the objective and analytical gradient, including the construction of the Gram matrix and the solution of $K\times K$ linear
systems.
The dominant complexity of one channel evaluation and update is $\mathcal O(MK^2+K^3)$.
Let $N_{\rm bt}$ denote the maximum number of backtracking trials per
AG iteration. For each rotation budget, AG-ZF performs at most
$I_{\rm AG}$ AG iterations. Accordingly, its overall complexity is
$\mathcal O\!\left(
Q_{\rm tot}
I_{\rm AG}N_{\rm bt}
(MK^2+K^3)
\right).$
Therefore, AG-ZF replaces the repeated SOCP solves required by DA-AO with analytical-gradient updates and structured matrix operations, thereby reducing the computational complexity.
\section{Simulation Results}
\label{SimulationResults}
In this section, numerical simulations are conducted to evaluate the proposed rotation-delay-aware RA-enabled UAV downlink communication scheme and the corresponding optimization algorithms.
Unless otherwise specified, the simulation parameters are summarized in Table~\ref{tab:simulation_parameters}.
The RAs rotate independently and in parallel at a common angular speed of \(\omega=2\pi\) rad/s, which is comparable to the maximum yaw tracking speed reported for the Trace-G1 antenna-tracking gimbal~\cite{ViewproTraceG1}. Since a higher rotation speed reduces the rotation delay, this setting provides a relatively conservative evaluation of the impact of rotation delay. Unless otherwise specified, the reported average results are obtained over $100$ independent user-location realizations.

\begin{table}[!t]
	\centering
		\caption{Simulation Parameters}
		\label{tab:simulation_parameters}
        \setlength{\tabcolsep}{2.5pt}
		\begin{tabular}{|c|c|c|}
			\hline
			\textbf{Parameter} & \textbf{Description} & \textbf{Value} \\
			\hline
			$f_c$ 
			& Carrier frequency 
			& $5.8$ GHz \\
			\hline
			
			$B$ 
			& System bandwidth 
			& $1$ MHz \\
			\hline
			
			$P_{\max}$ 
			& Maximum transmit power 
			& $10$ mW \\
			\hline
			
			$N_0$ 
			& Noise power spectral density 
			& $-174$ dBm/Hz \\
			\hline
			
			-- 
			& Receiver noise figure 
			& $7$ dB \\
			\hline
			
			$d$ 
			& Inter-element spacing
			& $\lambda/2$ \\
			\hline
			
			$M$ 
			& Number of transmit RAs 
			& $8$ ($2\times4$ UPA) \\
			\hline
			
			$K$ 
			& Number of users 
			& $4$ \\
			\hline
			
            $\mathbf p$
            & Transmitting UAV position
            & $(0,0,100)$ m \\
            \hline
            
            $u_{k,x},\,u_{k,y}$
            & Horizontal coordinates of users
            & $\mathcal U[-100,100]$ m \\
            \hline
            
            $u_{k,z}$
            & Vertical coordinates of users
            & $\mathcal U[0,100)$ m \\
            \hline
			
			$T_s$ 
			& Communication slot duration 
			& $0.5$ s \\
			\hline
			
			$\theta_{\max}$ 
			& Maximum mechanical deflection angle 
			& $\pi/2$ \\
			\hline
			
			$p$ 
			& Antenna directionality parameter 
			& $2$ \\
			\hline
			
			$\omega$ 
			& RA angular speed 
			& $2\pi$ rad/s \\
			\hline

   			$I_{\rm AG}$
			& Maximum number of AG iterations
			& $60$ \\
			\hline
			
			$\epsilon_{\rm AG}$
			& AG-ZF convergence tolerance
			& $10^{-5}$ \\
			\hline
   
   			$Q_0$
			& Number of initial rotation-budget samples
			& $5$ \\
			\hline
            $S$
            & Number of refinement stages
            & $2$ \\
            \hline
		\end{tabular}
\end{table}

\subsection{Single-User Performance}

\begin{figure}[!t]
	\centering
	\includegraphics[
	width=\columnwidth,
	trim=10 20 0 0,
	]{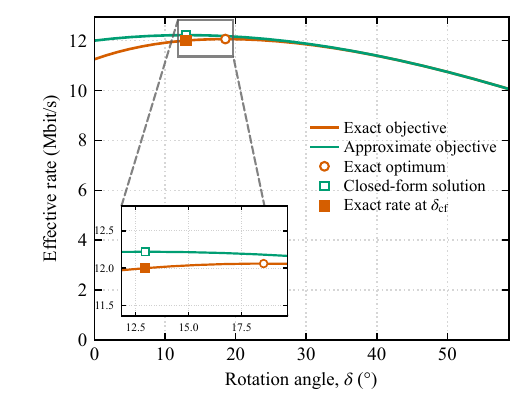}
	\caption{Effective transmission rate versus rotation budget for a representative single-user realization.}
	\label{fig:single_case}
\end{figure}

For the single-user system, we consider the following four schemes for comparison.
\textbf{Exact Solution}: the globally optimal rotation budget is obtained by
solving problem~\eqref{prob:far_field}.
When the sufficient condition in Theorem~\ref{thm:single_optimal_rotation}
is satisfied, the boundary cases are first identified using
\eqref{eq:optimal_delta_exact}, and bisection is applied when the optimum
lies in the interior.
Otherwise, a one-dimensional search is performed.
\textbf{Closed-form}: the rotation budget is determined by the proposed
second-order closed-form solution in \eqref{eq:delta_cf}.
\textbf{Delay-unaware}: the RA orientations are optimized without
accounting for the mechanical rotation delay, while the actual rotation
time is included when evaluating the effective transmission rate.
\textbf{Fixed Antenna}: all RAs remain at the common reference orientation $\mathbf f_0$.

We first evaluate the closed-form solution derived from the second-order
expansion.
Fig.~\ref{fig:single_case} compares the exact effective transmission rate and its second-order approximation as functions of the rotation budget $\delta$ for a representative single-user realization.
Here, $\delta^\star$ denotes the exact optimal rotation budget, whereas $\delta_{\rm cf}$ denotes the closed-form solution in \eqref{eq:delta_cf}.
Although $\delta_{\rm cf}$ does not exactly coincide with $\delta^\star$, the exact effective transmission rate achieved at $\delta_{\rm cf}$ is very close to the optimum.
This result demonstrates that the proposed closed-form solution can achieve a near-optimal effective transmission rate despite the difference in the selected rotation budget.

\begin{figure}[!t]
    \centering

    \subfloat[Varying $M$]{%
        \includegraphics[width=\columnwidth,trim=0 5 0 0,clip,]
        {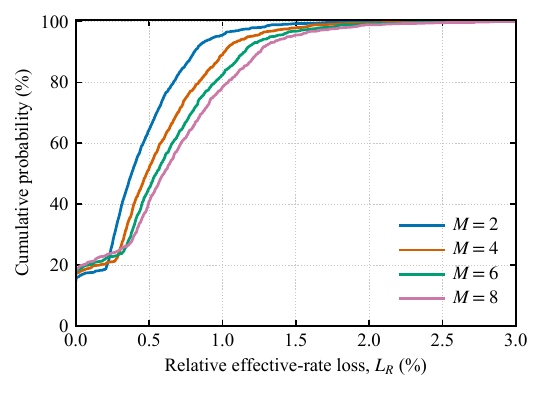}%
        \label{fig:single_loss_cdf_a}
    }

    \subfloat[Varying SNR]{%
        \includegraphics[width=\columnwidth,trim=0 5 0 0,clip,]
        {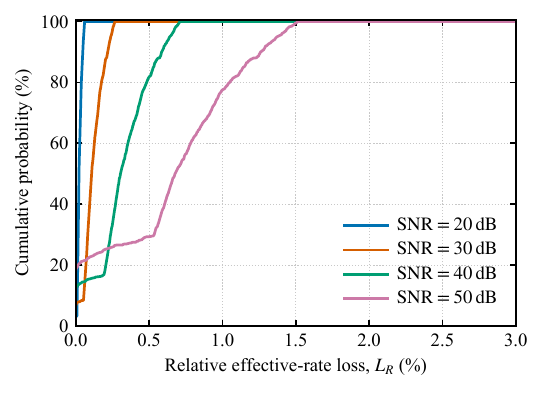}%
        \label{fig:single_loss_cdf_b}
    }

    \caption{CDFs of the relative effective-rate loss
    of the closed-form rotation design under different
    system configurations:
    (a) varying the number of RAs $M$;
    (b) varying the SNR with $M=4$.}

    \label{fig:single_loss_cdf}
\end{figure}

To quantify the performance gap between a test scheme and a reference scheme, we define the relative effective-rate loss as
\begin{equation}
    L_R\triangleq
    \frac{R_{\mathrm{ref}}-R_{\mathrm{test}}}{R_{\mathrm{ref}}}
    \times100\%,
    \label{eq:relative_rate_loss}
\end{equation}
where $R_{\mathrm{ref}}$ and $R_{\mathrm{test}}$ denote the effective rates of the reference and test schemes, respectively.
For the closed-form rotation design, $R_{\mathrm{ref}} =\bar R_{\mathrm{FF}}(\delta^\star)$ and
$R_{\mathrm{test}}=\bar R_{\mathrm{FF}}(\delta_{\mathrm{cf}})$.
Fig.~\ref{fig:single_loss_cdf} plots the cumulative distribution functions (CDFs) of $L_R$ for different array sizes and peak-aligned SNR values. 
As shown in Fig.~\ref{fig:single_loss_cdf}\subref{fig:single_loss_cdf_a}, the distribution of $L_R$ shifts toward larger values as $M$ increases. Similarly, Fig.~\ref{fig:single_loss_cdf}\subref{fig:single_loss_cdf_b} shows that $L_R$ increases with SNR. Overall, although the closed-form approximation does not necessarily yield the exact optimal rotation budget, it consistently achieves near-optimal effective-rate performance across the considered system configurations.

\begin{figure}[t]
    \centering
    \includegraphics[
    width=\columnwidth,
    trim=10 20 0 0,
    ]{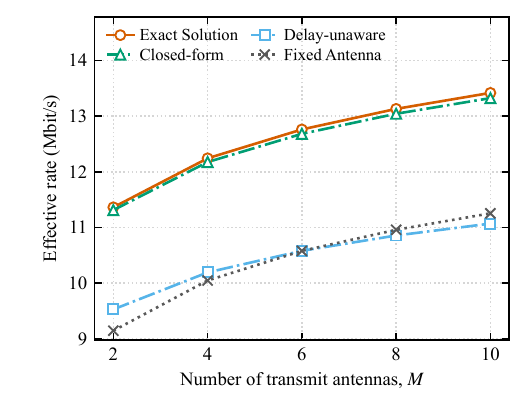}
    \caption{Average effective transmission rate versus the number of transmit RAs $M$ in the single-user system.}
    \label{fig:single_user_M}
\end{figure}

\begin{figure}[t]
    \centering
    \includegraphics[
    width=\columnwidth,
    trim=10 20 0 0,
    ]{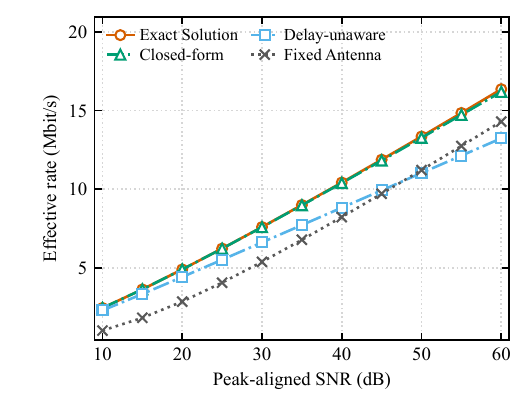}
    \caption{Average effective transmission rate versus the peak-aligned SNR in the single-user system.}
    \label{fig:single_user_snr}
\end{figure}

Figs.~\ref{fig:single_user_M} and~\ref{fig:single_user_snr} compare the average effective transmission rates versus the number of RAs $M$ and the peak-aligned SNR, respectively.
For both parameters, the exact and closed-form designs achieve the highest rates with a consistently small performance gap.
The delay-unaware design incurs a rate penalty because it selects its rotation angles without considering the resulting loss of transmission time.
Notably, Fixed Antenna overtakes the delay-unaware design at larger $M$ in Fig.~\ref{fig:single_user_M} and at sufficiently high SNR in Fig.~\ref{fig:single_user_snr}.
This observation shows that optimizing the post-rotation channel alone can be less effective than avoiding mechanical rotation altogether when its time cost is ignored.

Finally, Fig.~\ref{fig:single_regions} illustrates the impact of the mechanical rotation speed on the optimal rotation decision.
The mechanical limit is set to $\theta_{\max}=50^\circ$, and the angular speed $\omega$ is varied to change the normalized full-range rotation-time ratio $\rho=\theta_{\max}/(\omega T_s)$.
As $\rho$ increases, i.e., as the RAs rotate more slowly, the optimal solution successively transitions through three regimes.
Specifically, full rotation is optimal for $\rho\leq\rho_1$, where $\delta^\star=\bar{\delta}$; partial rotation is optimal for $\rho_1<\rho<\rho_2$, where $0<\delta^\star<\bar{\delta}$; and no
rotation is optimal for $\rho\geq\rho_2$, where $\delta^\star=0$.
When rotation is no longer beneficial, the effective rate levels off because the optimal design incurs no mechanical rotation delay.

\begin{figure}[t]
    \centering
    \includegraphics[
        width=\columnwidth,
        trim=10 10 10 0,
    ]{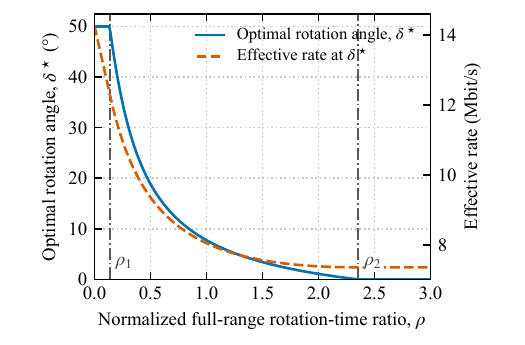}
    \caption{Optimal rotation budget and effective transmission rate versus the normalized full-range rotation time $\rho$ in the single-user system.}
    \label{fig:single_regions}
\end{figure}

\subsection{Multi-user Performance}

We next evaluate the performance of the proposed AG-ZF algorithm in multi-user scenarios. 
For comparison, the following three benchmark schemes are considered.
\textbf{Delay-Aware Alternating Optimization}:
for each rotation budget, the RA orientations and transmit beamforming are alternately optimized, and the effective rate accounting for rotation delay is used to select the budget.
\textbf{Delay-Unaware Alternating Optimization}:
the RA orientations and transmit beamforming are alternately optimized without considering the mechanical rotation delay. The actual rotation delay is included only when evaluating the effective transmission rate.
\textbf{Fixed Antenna}:
all directional antennas remain at the reference orientation $\mathbf f_0$,
and only the transmit beamforming is optimized.

\begin{figure}[t]
	\centering
	\includegraphics[
    width=\columnwidth,
    trim=10 25 0 12,
    ]{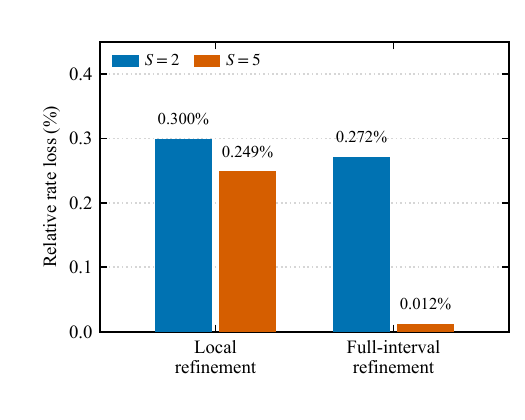}
	\caption{
		Average relative effective-rate loss under different outer-search
		refinement strategies and numbers of refinement stages $S$.
	}
	\label{fig:refinement}
\end{figure}

Fig.~\ref{fig:refinement} compares the relative effective-rate loss of AG-ZF under different outer-search refinement strategies and numbers of refinement stages $S$.
The loss is evaluated according to \eqref{eq:relative_rate_loss}, where $R_{\mathrm{ref}}$ is  obtained through a dense search over the entire feasible rotation-budget interval with a spacing of $0.25^\circ$, and $R_{\mathrm{test}}$ is the effective transmission rate achieved by the
corresponding refinement strategy.
With two refinement stages, local refinement incurs an average loss of $0.300\%$, compared with $0.272\%$ for full-interval refinement.
Increasing the number of stages to $S=5$ reduces the respective losses to $0.249\%$ and $0.012\%$.
Although deeper full-interval refinement yields the smallest loss, it requires additional budget evaluations. Since local refinement with $S=2$ achieves an average rate loss of only $0.3\%$, this configuration is adopted for the remaining simulations.

\begin{figure}[t]
	\centering
	\includegraphics[
	width=\columnwidth,
	trim=10 20 0 0,
	]{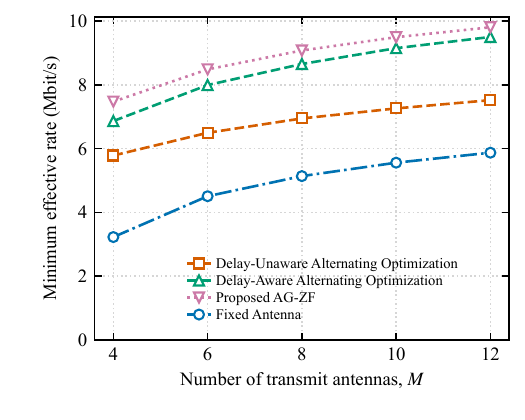}
	\caption{Average minimum effective transmission rate of different
		schemes versus the number of transmit RAs $M$.}
	\label{fig:multiuser_M}
\end{figure}

Fig.~\ref{fig:multiuser_M} compares the average minimum effective transmission rates of different schemes versus the number of RAs $M$.
The performance of all schemes improves as $M$ increases, owing to the higher array gain and enhanced spatial separation capability of the larger transmit array.
Under the same simulation settings, the proposed AG-ZF achieves a higher effective transmission rate than Delay-Aware Alternating Optimization over the considered range of $M$.
This result indicates that the proposed Gram-matrix-based optimization can achieve better-performing feasible solutions under the considered settings.
Delay-Unaware Alternating Optimization achieves lower rates because it ignores rotation delay during optimization, while Fixed Antenna achieves the lowest effective rate.

Fig.~\ref{fig:multiuser_K} further examines the impact of the number of users $K$ with the number of transmit RAs fixed.
As $K$ increases, separating the user channels becomes more challenging, leading to a decrease in the average minimum effective transmission rate for all three schemes. AG-ZF consistently achieves the highest effective rate, demonstrating the benefit of RA orientation optimization with mechanical rotation delay explicitly considered.

\begin{figure}[t]
    \centering
    \includegraphics[
    width=\columnwidth,
    trim=10 20 0 0,
    ]{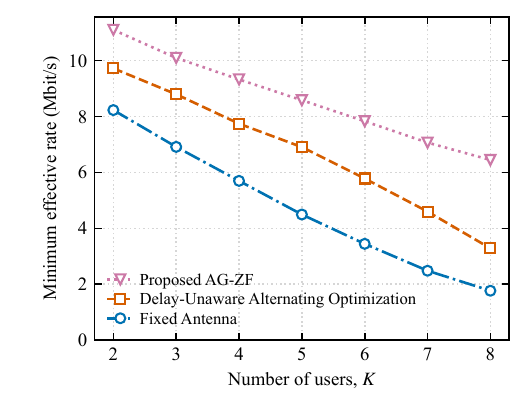}
    \caption{Average minimum effective transmission rate of different
    schemes versus the number of users $K$.}
    \label{fig:multiuser_K}
\end{figure}

Finally, Fig.~\ref{fig:multiuser_delay} examines the impact of the mechanical rotation speed on the average minimum effective transmission rate. With $\theta_{\max}$ fixed, the angular speed $\omega$ is varied to adjust $\rho$.
At small $\rho$, the proposed AG-ZF and Delay-Unaware Alternating Optimization achieve similar performance. As $\rho$ increases, their performance gap widens because the proposed algorithm limits the rotation extent to preserve transmission time, whereas the delay-unaware scheme continues to prioritize the post-rotation channel gain. At $\rho=0.8$, the proposed algorithm achieves approximately twice the effective rate of the delay-unaware scheme. Moreover, the latter eventually falls below Fixed Antenna, which incurs no rotation delay and thus maintains a constant effective rate. These results demonstrate the importance of explicitly accounting for mechanical rotation delay in RA orientation optimization.

\begin{figure}[t]
    \centering
    \includegraphics[
    width=\columnwidth,
    trim=10 20 0 0,
    ]{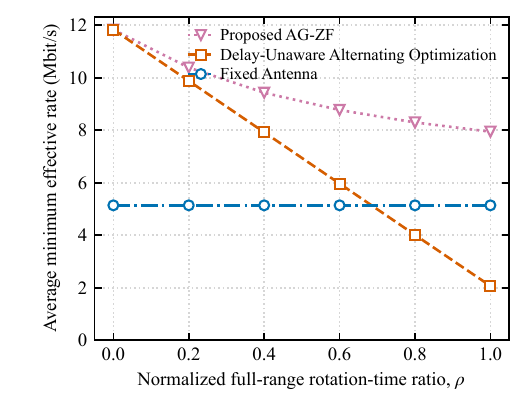}
    \caption{Average minimum effective transmission rate of different
    schemes versus the normalized full-range rotation time $\rho$.}
    \label{fig:multiuser_delay}
\end{figure}

\section{Conclusion}

This paper incorporated mechanical rotation delay into RA-enabled UAV downlink communications and jointly optimized RA orientations and transmit beamforming to maximize the minimum effective transmission rate under finite communication slots. 
The results show that accounting for rotation delay is necessary to maximize the effective transmission rate.
For the single-user system, we characterized three optimal rotation regimes, no rotation, partial rotation, and full rotation, under a sufficient condition for strict concavity. 
A closed-form approximation of the optimal rotation budget was further derived using a second-order expansion. 
For the multi-user system, a ZF-based formulation was developed to reduce the fixed-budget design to a Gram-matrix-based RA orientation problem with max-min fair power allocation.
Based on the resulting formulation, the AG-ZF algorithm was developed to efficiently solve the problem. 
Simulation results demonstrated the near-optimal effective-rate performance of the single-user closed-form approximation and the performance advantages of AG-ZF over the considered benchmark schemes. 

Future work may consider inter-slot RA reconfiguration under UAV mobility and incorporate practical factors such as rotation acceleration, settling time, mechanical energy consumption, and imperfect channel state information.
\appendices

\section{Proof of Theorem~\ref{thm:single_optimal_rotation}}
\label{app:proof_theorem1}
To prove the strict concavity of $\bar R_{\mathrm{FF}}(\delta)$, define
    $\ell(\delta)
    \triangleq
    \ln\!\left[
        1+C\cos^{2p}(\alpha-\delta)
    \right].$
Then,
    $\bar R_{\mathrm{FF}}(\delta)
    =
    \frac{B}{\ln 2}
    A(\delta)\ell(\delta).$

Over the feasible interval,
$0\leq\alpha-\delta<\pi/2$ since
$\bar{\delta}\leq\alpha<\pi/2$.
Differentiating $\ell(\delta)$ gives
\begin{equation}
	\ell'(\delta)
	=
	\frac{
		2pC\cos^{2p}(\alpha-\delta)
		\tan(\alpha-\delta)
	}{
		1+C\cos^{2p}(\alpha-\delta)
	}.
	\label{eq:appA_l_first}
\end{equation}
Thus, $\ell'(\delta)\geq0$, with
$\ell'(\delta)>0$ whenever $\delta<\alpha$.
The second derivative is
\begin{align}
	\ell''(\delta)
	={}&
	-
	\frac{
		2pC\cos^{2p}(\alpha-\delta)
	}{
		\left[
		1+C\cos^{2p}(\alpha-\delta)
		\right]^2
		\cos^2(\alpha-\delta)
	}
	\nonumber\\
	&\times
	\left[
		1
		+
		C\cos^{2p}(\alpha-\delta)
		-
		2p\sin^2(\alpha-\delta)
	\right].
	\label{eq:appA_l_second}
\end{align}
Under condition~\eqref{eq:concavity_condition},
it follows that $\ell''(\delta)\leq0$.
Moreover,
$A'(\delta)=-1/(\omega T_s)$ and $A''(\delta)=0$.
Therefore,
\begin{equation}
	\bar{R}_{\mathrm{FF}}''(\delta)
	=
	\frac{B}{\ln 2}
	\left[
	A(\delta)\ell''(\delta)
	-
	\frac{2}{\omega T_s}\ell'(\delta)
	\right].
	\label{eq:appA_rate_second}
\end{equation}

Since
$0\leq\delta\leq\bar{\delta}\leq\delta_{\max}\leq\omega T_s$,
we have $A(\delta)\geq0$.
Together with $\ell'(\delta)>0$ for $\delta<\alpha$ and
$\ell''(\delta)\leq0$, \eqref{eq:appA_rate_second} gives
$\bar{R}_{\mathrm{FF}}''(\delta)<0$ for $\delta<\alpha$.
Hence, if $\bar{\delta}<\alpha$,
$\bar{R}_{\mathrm{FF}}(\delta)$ is strictly concave over
$[0,\bar{\delta}]$.
If $\bar{\delta}=\alpha$, the same result holds on
$[0,\alpha)$, and continuity at the endpoint extends the
strict concavity to $[0,\alpha]$.
Thus, $\bar{R}_{\mathrm{FF}}(\delta)$ is strictly concave over
$[0,\bar{\delta}]$ in either case, and problem~\eqref{prob:far_field}
has a unique global optimum.

The location of the optimum follows directly from the strict monotonicity of
$\bar{R}_{\mathrm{FF}}'(\delta)$.
If $\bar{R}_{\mathrm{FF}}'(0)\leq0$, then
$\delta^\star=0$.
If $\bar{R}_{\mathrm{FF}}'(\bar{\delta})\geq0$, then
$\delta^\star=\bar{\delta}$.
Otherwise,
$\bar{R}_{\mathrm{FF}}'(0)>0$ and
$\bar{R}_{\mathrm{FF}}'(\bar{\delta})<0$,
so there exists a unique
$\delta_0\in(0,\bar{\delta})$
such that
$\bar{R}_{\mathrm{FF}}'(\delta_0)=0$.
Hence, $\delta^\star=\delta_0$.
Combining the three cases gives
\eqref{eq:optimal_delta_exact}, which completes the proof.

\section{Proof of Theorem~\ref{thm:closed_form_rotation}}
\label{app:proof_theorem2}

Define
$\psi(x)\triangleq\ln\!\left(1+C\cos^{2p}x\right)$.
Since $\psi(x)$ is an even function,
$\psi(0)=\ln(1+C)$,
$\psi'(0)=0$, and
$\psi''(0)=-\frac{2Cp}{1+C}$.
Its second-order Taylor expansion about $x=0$ is therefore
\begin{equation}
	\psi(x)
	=
	\ln(1+C)
	-
	\frac{Cp}{1+C}x^2
	+
	\mathcal O(x^4).
	\label{eq:appB_taylor}
\end{equation}

Retaining terms up to second order and substituting
\eqref{eq:appB_taylor} into \eqref{eq:rate_x} gives
\begin{equation}
	\widetilde{\mathcal R}(x)
	=
	\frac{B}{\omega T_s\ln2}
	(a+x)
	\left[
	\ln(1+C)
	-
	\frac{Cp}{1+C}x^2
	\right].
	\label{eq:appB_approx_rate}
\end{equation}

For notational simplicity, define
$\kappa\triangleq\frac{Cp}{1+C}$ and
$L\triangleq\ln(1+C)$.
Since $C>0$ and $p>0$, both $\kappa$ and $L$ are positive.
Then,
\begin{equation}
	\widetilde{\mathcal R}(x)
	=
	\frac{B}{\omega T_s\ln2}
	(a+x)(L-\kappa x^2).
	\label{eq:appB_compact_rate}
\end{equation}
Its first and second derivatives are, respectively,
\begin{align}
	\widetilde{\mathcal R}'(x)
	&=
	\frac{B}{\omega T_s\ln2}
	\left(
	L-2a\kappa x-3\kappa x^2
	\right),
	\label{eq:appB_first_derivative}
	\\
	\widetilde{\mathcal R}''(x)
	&=
	-\frac{2B\kappa}{\omega T_s\ln2}
	(a+3x).
	\label{eq:appB_second_derivative}
\end{align}

Since $a>0$, $\kappa>0$, and $x\geq0$ over the feasible interval,
\eqref{eq:appB_second_derivative} gives
$\widetilde{\mathcal R}''(x)<0$.
Hence, $\widetilde{\mathcal R}(x)$ is strictly concave, and its unique stationary point is the unconstrained maximizer.

Setting $\widetilde{\mathcal R}'(x)=0$ gives
\begin{equation}
	3\kappa x^2
	+
	2a\kappa x
	-
	L
	=
	0.
	\label{eq:appB_stationary}
\end{equation}
The two roots are
\begin{equation}
	x_{\pm}
	=
	\frac{
		-a
		\pm
		\sqrt{
			a^2+\frac{3L}{\kappa}
		}
	}{3}.
	\label{eq:appB_roots}
\end{equation}
Since $L>0$ and $\kappa>0$,
$\sqrt{a^2+3L/\kappa}>a$.
Thus, $x_-<0$ and $x_+>0$, so the unique nonnegative stationary point is
\begin{equation}
	x_{\mathrm{cf}}
	=
	\frac{
		\sqrt{
			a^2+\frac{3L}{\kappa}
		}
		-a
	}{3}.
	\label{eq:appB_xcf}
\end{equation}
Substituting $L=\ln(1+C)$ and
$\kappa=Cp/(1+C)$ gives \eqref{eq:x_cf}.

Since $\widetilde{\mathcal R}(x)$ is strictly concave, $x_{\mathrm{cf}}$ is optimal if it lies in the feasible interval
$[\alpha-\bar{\delta},\alpha]$; otherwise, the optimum is attained at the corresponding boundary point.
Using $\delta=\alpha-x$, the resulting rotation budget is given by
\eqref{eq:delta_cf}, which completes the proof.
\bibliographystyle{IEEEtran}
\bibliography{refs}

\end{document}